\documentclass[11pt]{article}

\usepackage{amsmath,amsthm}
\usepackage{mathtools}
\usepackage{fullpage}

\usepackage{url,hyperref}
\usepackage{xcolor,colortbl}

\def\etal{{\it et~al.}\,}

\newcolumntype{C}{>{$\ }c<{\ $}} 

\newcommand{\hcell}[1]{\cellcolor{yellow}#1} 

\newtheorem{theorem}{Theorem}
\usepackage{algorithmic}
\usepackage{letltxmacro}
\LetLtxMacro{\oldalgorithmic}{\algorithmic}
\LetLtxMacro{\endoldalgorithmic}{\endalgorithmic}
\renewenvironment{algorithmic}[1][0]{%
  \hrulefill\par
  \oldalgorithmic[#1]}
  {\endoldalgorithmic\par
   \vspace*{-.5\baselineskip}
   \hrulefill\par
  }

\begin{document}

\date{}

\title{Computing random $r$-orthogonal Latin squares}
\author{Sergey Bereg\thanks{
Department of Computer Science,
University of Texas at Dallas,
Box 830688,
Richardson, TX 75083
USA.}
}

\maketitle

\vspace{-5mm}

\begin{abstract}
Two Latin squares of order $n$ are $r$-orthogonal if, when superimposed, there are exactly $r$ distinct ordered pairs.
The spectrum of all values of $r$ for Latin squares of order $n$ is known. 
A Latin square $A$ of order $n$ is $r$-self-orthogonal if $A$ and its transpose are $r$-orthogonal. 
The spectrum of all values of $r$ is known for all orders $n\ne 14$. 
We develop randomized algorithms for computing pairs of $r$-orthogonal Latin squares of 
order $n$ and algorithms for computing $r$-self-orthogonal Latin squares of order $n$.
\end{abstract}

\section{Introduction}

For a positive integer $n$,
a Latin square of order $n$ is an $n\times n$ array
filled with $n$ different symbols, each occurring exactly 
once in each row and exactly once in each column. 
For a positive integer $n$, let $[n]$ denote the set $\{1,2,\dots,n\}$.
We denote the $(i,j)$-th element of a Latin square $A$ by $A_{i,j}$.
In this paper, we assume that the symbols in a Latin square of order $n$ are in $[n]$, i.e. $A_{i,j}\in [n]$ for all $i,j\in [n]$.

Two Latin squares $A$ and $B$ of order $n$ 
are said to be {\em orthogonal} if every ordered pair of symbols 
occurs exactly once among the $n^2$ pairs $(A_{i,j},B_{i,j}), i\in [n],j\in [n]$.
Orthogonal Latin squares were studied in detail by Leonhard Euler see Fig. \ref{f1}.  
Orthogonal Latin squares exist for all orders $n>1$ except $n = 2, 6$.

\begin{figure}[htb]
\centering
\begin{tabular}{|C|C|C|C|}
\hline
A\alpha & B\delta & C\beta & D\gamma\\\hline
B\gamma & A\beta & D\delta & C\alpha\\\hline
C\delta & D\alpha & A\gamma & B\beta\\\hline
D\beta & C\gamma & B\alpha & A\delta \\\hline
\end{tabular}
\caption{Orthogonal Latin squares using Euler's notation (one Latin square uses the first $n$ upper-case letters from the Latin alphabet
 and the other uses the first $n$ lower-case letters from the Greek alphabet). 
Orthogonal Latin squares are also known as  Graeco-Latin squares or Euler squares. }
\label{f1}
\end{figure}

If $A$ and $B$ are orthogonal Latin squares, then we say that $A$
is an {\em orthogonal mate} of $B$ and vice versa.
In 1779 Euler \cite{euler23} proved that for every even $n$ 
there exists a Latin square of order $n$
that has no orthogonal mate.
van Rees \cite{rees90} called such squares {\em bachelor} Latin squares. 
The existence of bachelor Latin squares of orders $n\equiv 1\pmod 4$ was shown by Mann \cite{mann44}.
In 2006 Evans \cite{evans06} and Wanless and Webb \cite{ww06}
completed the remaining case of Latin squares of orders $n\equiv 3\pmod 4$.

\begin{theorem}[\cite{evans06,ww06}] 
For any positive integer $n\notin\{1,3\}$ 
there exists a Latin square of order $n$ that
has no orthogonal mate.
\end{theorem}

van Rees \cite{rees90} has conjectured that the portion of bachelor Latin squares of 
order $n$ tends to one as $n\to\infty$. 
However, this conjecture was based on a study of Latin squares of
small orders. 
Computational results by Egan and Wanless \cite{egan12} show 
that only a small proportion of Latin squares of orders 7, 8, and 9 possess orthogonal mates. 
McKay \etal \cite{mckay07} tested 10 million
random Latin squares of order 10 and estimated that
around 60\% of Latin squares of order 10 have mates.

A set of Latin squares of the same order, all pairs of which are orthogonal 
is called a set of {\em mutually orthogonal Latin squares} (MOLS). 
The study of mutually orthogonal Latin squares is a subject that has attracted
much attention due to applications in error correcting codes, 
cryptographic systems, compiler testing,
and statistics (the design of experiments) \cite{keedwell15}. 

The construction of MOLS is a notoriously difficult 
combinatorial problem, see for example \cite{mfl-cols-16,mgfl-20}.
In this paper, we study the problem of computing $r$-orthogonal 
Latin squares.
Two Latin squares of order $n$ are {\em $r$-orthogonal} ($r$-OLS)
if their superposition produces exactly $r$
distinct ordered pairs. 
Thus, orthogonal Latin squares of order $n$ are $n^2$-orthogonal. 
Belyavskaya \cite{bel76,bel77,bel92} systematically treated the following
question: For which integers $n$ and $r$ does a pair of 
$r$-orthogonal Latin squares of order $n$ exist? 
The spectrum of $r$-orthogonal Latin squares of order $n$ was determined by 
Colbourn and Zhu \cite{cz-srols-95} (with few exceptions) and 
completed later by Zhu and Zhang~\cite{zz-cs-03}.

\begin{theorem}[Zhu and Zhang~\cite{zz-cs-03}] 
\label{thm-pairs}
For any integer $n\ge 2$, there exists a pair of $r$-orthogonal Latin
squares of order $n$ if and only if $n\le r\le n^2$ and $r\notin\{n+1, n^2-1\}$ 
with the exceptions of $n$ and $r$ shown in Table \ref{tableT1T2}(a).
\end{theorem}

\begin{table}[htb]
\centering
\begin{tabular}{ccc}

\begin{tabular}{|c|l| }
\hline
$n$ & Genuine exceptions of $r$\\\hline
 2 & 4  \\ 
 3 & 5, 6, 7  \\  
 4 & 7, 10, 11, 13, 14\\
5 & 8, 9, 20, 22, 23\\
6 & 33, 36\\\hline
\end{tabular}

& \quad\quad & 

\begin{tabular}{|c|l| }
\hline
$n$ & Genuine exceptions of $r$\\\hline
 2 & 4  \\ 
 3 & 5, 6, 7, 9  \\  
 4 & 6, 7, 8, 10, 11, 12, 13, 14\\
5 & 8, 9, 12, 16, 18, 20, 22, 23\\
6 & 32, 33, 34, 36\\
7 & 46\\\hline
\end{tabular}

\\

(a) && (b) 

\end{tabular}
\vspace{3mm}
\caption{(a) Genuine exceptions of pairs of $r$-orthogonal Latin squares of order $n$. 
(b) Genuine exceptions of $r$-self-orthogonal Latin squares of order $n$.
}
\label{tableT1T2}
\end{table}

A Latin square which is orthogonal to its transpose is called a
{\em self-orthogonal Latin square} (SOLS). It is known that self-orthogonal Latin
squares exist for all orders $n\notin\{2,3,6\}$\cite{fz-sols-07}.
We say that a Latin square $A$ is {\em $r$-self-orthogonal} ($r$-SOLS) 
if $A$ and its transpose $A^T$ are $r$-orthogonal.
The spectrum of $r$-SOLS has been almost completely decided by Xu and
Chang \cite{zz-cs-03} and Zhang \cite{z-25-13}, as shown in the following theorem.

\begin{theorem}[Zhang~\cite{z-25-13}] \label{thm-self}
For any integer $n\ge 1$, there exists an $r$-self-orthogonal Latin square of order $n$ if and only if 
$n\le r\le n^2$ and $r\notin\{n+1, n^2-1\}$ with 26 genuine exceptions of $n$ and $r$ shown in Table \ref{tableT1T2}(b)
and one possible exception of $(n,r)=(14,14^2-3)$.
\end{theorem}

In this paper, we develop randomized algorithms for computing $r$-orthogonal Latin squares
and $r$-self-orthogonal Latin squares.
The benefit of this approach is that $r$-orthogonal Latin squares 
of large orders can be computed. 
This is due to the polynomial running time. 
Our experiments show that the difficult case is when $r$ is close to $n^2$.
Evidence of this is the open problem of finding a $193$-self-orthogonal Latin 
square of order 14 \cite{z-25-13} which is the only case left to complete 
the spectrum of $r$-self-orthogonal Latin squares.  

{\em Related work}. 
Keedwell and Mullen \cite{km-spols-04}
investigated the construction of sets of $t$ Latin squares of a given non-prime-power order $q$ which are as close as possible to being a mutually orthogonal set. 
Dinitz and Stinson \cite{ds-mndop-05}
studied the problem of constructing sets of $s$ Latin squares of order $m$
such that the average number of different ordered pairs obtained by superimposing two of the $s$
squares in the set is as large as possible.
Arce-Nazario \etal \cite{acchmr-14}
discussed some computational problems concerning the distribution of orthogonal
pairs in sets of Latin squares of small orders.

\section{Preliminaries}

A {\em partial Latin square} of order $n$ is an $n\times n$ array in which\\
(i) each entry is either empty or it contains an element from $[n]$, and\\
(ii) each of the symbols $1,2,\dots,n$ occurs at most once in each row and at most once in each column of the array.
Completing partial Latin squares is NP-complete \cite{c-ccpls-84}.
However, some partial Latin squares can always be completed, for example, Latin rectangles. 
A Latin rectangle is a $k\times n$ array $A$ with $k\le n$ and $A_{i,j}\in [n]$ such that 
no number occurs more than once in any row or column.

\begin{theorem}[M. Hall \cite{hall1945}] \label{hall45}
Every $r\times n$ Latin rectangle, $0\le r\le n$, can be completed to a Latin square of order $n$. 
\end{theorem}

The proof is based on the following theorem for SDR. 
A {\em system of distinct representatives}, or {\em SDR}, for a collection of finite sets
$S_1,S_2,\dots,S_m$, is a sequence $\langle s_1,s_2,\dots,s_m\rangle$
of $m$ distinct elements $s_i\in S_i$. 
Each $s_i$ is called a {\em representative} of set $S_i$. 
For example, the set $\langle 2,3,1\rangle$ is an SDR for the sets $S_1=\{2\}$, 
$S_2=\{1,3\}$ and $S_3=\{1,2\}$.

\begin{theorem}[P. Hall \cite{hall1935}] \label{hall35}
Let $S_1,S_2,\dots,S_m$ be a collection of $m$ finite sets.
Then an SDR for these sets exists if and only if, for all $k\in \{0,1,\dots,m\}$,
$|S_{i_1}\cup S_{i_2}\cup\dots\cup S_{i_k}|\ge k$, where the $k$ sets 
$S_{i_1},\dots,S_{i_k}$ represent any collection of $k$ sets chosen from the $m$ sets 
$S_1,S_2,\dots,S_m$.
\end{theorem}

The proof of Theorem \ref{hall45} uses a collection of sets $S_i,i\in [n]$,
where $S_i$ is the set of all $x\in [n]$ such that $x$ does not
occur in the column $j$ of a $r\times n$ Latin rectangle $A$. 
It can be shown that these sets satisfy Hall's condition (Theorem \ref{hall35}):  
for all $k\in \{0,1,\dots,n\}$, the union of any $k$ sets in the collection contains at least $k$ elements.
By Theorem \ref{hall35}, an SDR for these sets exists.
Therefore Latin rectangle $A$ can be extended to a $(r+1)\times n$ Latin rectangle by adding 
a row of the representatives. 
An SDR can be computed as a matching in a bipartite graph $G_A=(V_1,V_2,E)$ 
where $V_1=\{v_1,\dots,v_n\}, V_2=\{u_1,\dots,u_n\}, E=\{(v_i,u_j)~|~j\in S_i\}$.
A Latin square can be computed using this step $n-r$ times.

\section{Completing Orthogonal Latin Rectangles}

In order to construct $r$-orthogonal Latin squares of order $n$, 
we apply a randomized algorithm called Algorithm A1 where
\begin{enumerate}
\item
the first rows of two Latin squares $A$ and $B$ are 
computed as random permutations of $[n]$ and
\item
the remaining rows of $A$ and $B$ are computed using a matching algorithm 
applied to bipartite graphs $G_{A'}$ and $G_{B'}$
and random order of sets $S_i$ where $A'$ and $B'$ are the current Latin rectangles. 
\end{enumerate} 

\begin{figure}[htb]
\centering
$A=$
\begin{tabular}{|ccccccc|}
\hline
 1 & 2 & 7 & 5 & 6 & 4 & 3\\
 6 & 7 & 4 & 2 & 3 & 1 & 5\\
 3 & 1 & 6 & 7 & 2 & 5 & 4\\
 5 & 4 & 3 & 1 & 7 & 6 & 2\\
 4 & 5 & 2 & 3 & 1 & 7 & 6\\
 2 & 6 & 1 & 4 & 5 & 3 & 7\\
 7 & 3 & 5 & 6 & 4 & 2 & 1\\
\hline
\end{tabular}
\qquad
$B=$
\begin{tabular}{|ccccccc|}
\hline
 4 & 2 & 5 & 3 & 7 & 1 & 6\\
 3 & 6 & 2 & 4 & 5 & 7 & 1\\
 7 & 3 & 6 & 2 & 1 & 5 & 4\\
 6 & 1 & 7 & 5 & 4 & 2 & 3\\
 2 & 7 & 3 & 1 & 6 & 4 & 5\\
 5 & 4 & 1 & 6 & 2 & 3 & 7\\
 1 & 5 & 4 & 7 & 3 & 6 & 2\\
\hline
\end{tabular}
\caption{A pair of $42$-orthogonal Latin squares of order 7 where all pairs $(i,j)$ 
except $(2,7),(3,2),(3,4),(4,5),(4,7),(6,1),(7,3)$ appear
if $A$ and $B$ are superimposed.}
\label{t42}
\end{figure}

We implemented Algorithm A1 and run it for $n\in [5,20]$, see for example $42$-orthogonal Latin squares of order $7$
in Fig. \ref{t42}. 
The range of values of $r$ for $n\in [5,20]$ computed by Algorithm A1 
is shown in Table \ref{t:res}. 
One can observe that the range is complete for $n=5$ and $n=6$ (by Theorem \ref{thm-pairs}).
For $n\ge 7$, the algorithm only found a subset of possible values of $r$ for two $r$-orthogonal Latin squares. 
For example, it computed 36 values of $r$ for $n=10$ which is $\frac{33}{89}\cdot 100\approx 37.07\%$ of all values. 
For $n=20$, it computed 70 values of $r$ which is $\frac{70}{379}\cdot 100\approx 18.46\%$ of all values. 
So, one needs to develop an algorithm for computing $r$-orthogonal Latin squares of order $n$ for large and small values of $r$.
It also gives rise to an interesting problem.

\begin{quote}
{\bf Open problem}. Let $r(A,B)$ be the {\em orthogonality} of two Latin squares $A$ and $B$ of order $n$, i.e.
the number of distinct ordered pairs when $A$ and $B$ are superimposed. 
What is the expected value of $r(A,B)$ of two random Latin squares $A$ and $B$ of order $n$?
\end{quote}

We computed approximately the expected value of $r(A,B)$ of two random Latin squares $A$ and $B$ 
using the program\footnote{It is based on the Java implementation described by Ignacio Gallego Sagastume\\
\href{https://github.com/bluemontag/igs-lsgp}{https://github.com/bluemontag/igs-lsgp}.} 
developed by Paul Hankin \cite{paulhankin19} implementing an algorithm for generating random Latin squares
by Jacobson and Matthews \cite{jm-gud-96}. 
10000 runs were used for $n=5,\dots,12$ and 1000 runs were used for $n=13,\dots,20$.
The result is shown in Table \ref{t:res} (column $r(A,B)$).
For all $n\in\{5,6,\dots,20\}$, the expected value of $r(A,B)/n^2$ is close to 0.63.


In order to extend the range of $r(A,B)$ computed by Algorithm A1, 
a plausible approach would be to solve the problem of extending two $k\times n$ 
Latin rectangles to two $(k+1)\times n$ Latin rectangles maximizing (minimizing) $r(A,B)$.
This problem seems difficult and we consider an approach where one Latin rectangle
is extended first and then the other Latin rectangle is extended by solving 
the following problem. 

{\bf Problem MaxRAB} ({\bf MinRAB}). 
Let $A$ be a $k\times n$ Latin rectangle, $k<n$, and let $B$ be a $(k+1)\times n$ Latin rectangle.
Extend $A$ to a $(k+1)\times n$ Latin rectangle $A'$ maximizing (resp. minimizing) $r(A',B)$.

\newpage

\begin{table} 
\centering
\begin{tabular}{|r||c|c|c|} 
\hline
$n$ & A1 & A2 & $r(A,B)$\\
\hline
\hline
5 & 5,7,10-19,21,25 & - & 15.8478\\\hline
6 & 6,8-32,34 & - & 22.7416\\\hline
7 & 13-44 & 7,9,11-47,49 & 30.9792\\\hline
8 & 25-53 & 8,10-60 & 40.4537\\\hline
9 & 35-65 & 9,11-77 & 51.2081\\\hline
10 & 47-79 & 10,12-93 & 63.2259\\\hline
11 & 60-93 & 11,13-114 & 76.55\\\hline
12 & 71-109 & 12,14-135 & 91.0366\\\hline
13 & 85-128 & 13,15-157 & 106.847\\\hline
14 & 102-145 & 14,16-183 & 124.075\\\hline
15 & 117-166 & 15,17-210 & 142.279\\\hline 
16 & 137-189 & 16,18-236 & 161.437\\\hline
17 & 155-207,209 & 17,19-268 & 182.85\\\hline
18 & 173,176-233 & 18,20-300 & 204.693\\\hline
19 & 196-257,259,261 & 19,21-334 & 228.004\\\hline
20 & 218-286,290 & 20,22-370 & 252.536\\\hline 
\end{tabular}
\vspace{3mm}
\caption{
Pairs of $r$-orthogonal Latin squares of order $n$ computed by Algorithm A1 (column A1) and Algorithm A2 (column A2).
No entry in column A2 means the same result as in column A1.
Column $r(A,B)$ shows the average orthogonality of two random Latin squares using the program developed 
by Paul Hankin \cite{paulhankin19}.}
\label{t:res}
\end{table}

\begin{theorem} \label{tMaxRAB}
Problems MaxRAB and MinRAB can be solved by $O(n^3)$ time algorithm.
\end{theorem}

\begin{proof}
Let $A$ be a $k\times n$ Latin rectangle and $B$ be a $(k+1)\times n$ Latin rectangle.
Let $A_j$ denote the set of positive integers that do not occur in column $j$ of $A$.
Let $S=\{(A_{i,j},B_{i,j})~|~i\in [k], j\in [n]\}$.
Construct a weighted bipartite graph $G=(V_1,V_2,E)$ where 
$V_1 =\{v_1,\dots,v_n\}$, $V_2 =\{u_1,\dots,u_n\}$, and
$E =\{(v_i,u_j)~|~j\in A_i\}$.
To complete the construction, we assign the weight to each edge $(v_i,u_j)\in E$ as 
\[ 
w(v_i,u_j) = \begin{cases}
  1 & \text{if $(j,B_{k+1,i}) \notin S$,} \\
  0 & \text{if $(j,B_{k+1,i}) \in S$}.
\end{cases}
\]

Graph $G$ has a perfect matching since Latin rectangle $A$ can be extended to a 
$(k+1)\times n$ Latin rectangle by Theorem \ref{hall45}. 
In order to solve problem MaxRAB, 
we compute the maximum weight matching $M$ which is the solution of the assignment problem.
It can be found using the Hungarian method \cite{bc-lape-99,kuhn1955} in $O(n^3)$ time.
We extend Latin rectangle $A$ to a $(k+1)\times n$ Latin rectangle $A'$ by setting
$A'_{k+1,i}=j$ if $(v_i,u_j)$ is an edge of the matching.
It remains to prove that $r(A',B)$ is maximized.

Let $A^*$ be an extension of $A$ (i.e. $A^*$ is a $(k+1)\times n$ Latin rectangle such that the first $k$ rows of $A^*$ and $A$ are the same)
maximizing $r(A^*,B)$.
Consider the permutation $A^*_{k+1,1},A^*_{k+1,2},\dots,A^*_{k+1,n}$.
Clearly, $A^*_{k+1,i}\in A_i$ for all $i\in [n]$.
Therefore $M^*=\{(v_i,u_j)~|~i\in [n], j=A^*_{k+1,i}\}$ is a matching in graph $G$.
Let $B'$ be a $k\times n$ Latin rectangle obtained using the first $k$ rows of Latin rectangle $B$. 
Then $r(A^*,B) - r(A,B')$ is the number of pairs $(A^*_{k+1,i},B_{k+1,i})$ in $[n]^2\setminus S$.
Therefore $r(A^*,B) - r(A,B')=w(M^*)$.  
Similarly $r(A',B) - r(A,B')=w(M). $
Since $M$ is the solution of the assignment problem, $w(M)\ge w(M^*)$.
Then 
\begin{align*}
r(A',B) &= r(A,B')+w(M)\\
&\ge   r(A,B')+w(M^*)\\
&= r(A^*,B).
\end{align*}

Since $r(A^*,B)\ge r(A',B)$, we have $r(A^*,B) = r(A',B)$. So, Latin rectangle $A'$ is optimal. 

A similar argument can be used to solve problems MinRAB.
The minimum weight matching can be applied on the same graph $G$. 
\end{proof}

By Theorem \ref{tMaxRAB}, the extension of Latin rectangle $A$ is optimal.
However, the extension of Latin rectangle $B$ in our approach might be not the best.
We attempt to optimize (maximize or minimize) $r(A,B)$ using cycle switches \cite{am-tls-90,w-csls-04}.

\begin{figure}[htb]
\centering
\begin{tabular}{ccccc}

\begin{tabular}{|ccccc|}
\hline
 4&3&2&5&1\\
 5 & 1 & 3 & 2 & 4 \\
 3 & 4 & 5 & 1 & 2 \\
 1 & 2 & 4 & 3 & 5 \\
 2 & 5 & 1 & 4 & 3 \\
\hline
\end{tabular}

& \quad\quad & 

\begin{tabular}{|ccccc|}
\hline
4&3&2&5&1\\
 5 & 1 & 3 & 2 & 4 \\
 3 & \hcell{5} & \hcell{1} & \hcell{4} & 2 \\
 1 & 2 & 4 & 3 & 5 \\
 2 & \hcell{4} & \hcell{5} & \hcell{1} & 3 \\
\hline
\end{tabular}

& \quad\quad & 

\begin{tabular}{|ccccc|}
\hline
 4 & \hcell{2} & \hcell{3} & 5 & 1 \\
 5 & 1 & \hcell{2} & \hcell{3} & 4 \\
 3&4&5&1&2\\
 1 & \hcell{3} & 4 & \hcell{2} & 5 \\
 2 & 5 & 1 & 4 & 3 \\
\hline
\end{tabular}
\\

(a) && (b) && (c) 

\end{tabular}

\caption{(a) A Latin square $A$. 
(b) Switching a row cycle in $A$ between the third row and the fifth row.
(c) Switching a symbol cycle in $A$ on symbols 2 and 3.
}
\label{tableRS1}
\end{figure}

{\em Switching cycles}. 
Consider distinct rows $r$ and $s$ in a Latin square $A$.
Let $\pi_{r,s}$ be a permutation which maps $A_{r,i}$ to $A_{s,i}$.
Clearly, $\pi_{r,s}$ is a derangement, i.e. it has no fixed points. 
Take any cycle of $\pi_{r,s}$ and let $C$ be the set of columns involved in the cycle.
Switching {\em row cycle} $C$ in $A$ is defined by
\[ 
A'_{i,j} = 
\begin{cases}
  A_{s,j}  & \text{if $i=r$ and $j\in C$,} \\
  A_{r,j}  & \text{if $i=s$ and $j\in C$,} \\
  A_{i,j}  & \text{otherwise}, 
\end{cases}
\]
see an example in Fig. \ref{tableRS1}(b).
A {\em column cycle} is a set of elements that form a row cycle when the square is
transposed. 
Switching a {\em symbol cycle} on two symbols $a$ and $b$ is achieved by replacing every
occurrence of $a$ in the cycle by $b$ and vice versa, see an example in Fig. \ref{tableRS1}(c).

\begin{figure}[htb]
\centering
\begin{tabular}{ccccc}

\begin{tabular}{|ccccc|}
\hline
4&3&2&5&1\\
 5 & 1 & 3 & 2 & 4 \\
 3 & 4 & 5 & 1 & 2 \\
 1 & 2 & 4 & 3 & 5 \\
\hline
\end{tabular}

& \quad\quad & 

\begin{tabular}{|ccccc|}
\hline
 \hcell{2} & 3 & \hcell{4} & 5 & 1 \\
5&1&3&2&4 \\
 3 & 4 & 5 & 1 & 2 \\
 \hcell{4} & 2 & \hcell{1} & 3 & 5 \\
\hline
\end{tabular} 

& \quad\quad & 

\begin{tabular}{|ccccc|}
\hline
4&3&2&5&1\\
 5 & 1 & \hcell{4} & 2 & \hcell{3} \\
 3 & 4 & 5 & 1 & 2 \\
 1 & 2 & \hcell{3} & \hcell{4} & 5 \\
\hline
\end{tabular}
\\
(a) && (b) && (c) 
\end{tabular}
\caption{(a) A Latin rectangle $A$. 
(b) Function $f_{1,3}$ for $A$ induces a cycle (3,5) and a path (1,4,2). 
A Latin rectangle after switching column path (1,4,2) in $A$ is shown.
(c) Switching a symbol path in $A$ on symbols 3 and 4.
}
\label{tableRS2}
\end{figure}

In our approach, we have Latin rectangles instead of Latin squares.
Switching row cycles can be applied to Latin rectangles since the rows are full.
We adapt column cycles to Latin rectangles as {\em column paths}.
Consider distinct columns $c$ and $d$ in a Latin rectangle $A$.
Let $f_{c,d}$ be a function that maps $A_{i,c}$ to $A_{i,d}$.
Start with a symbol $A_{i,c}$ in column $c$ which is not present in column $d$.
Do the following step for row $i$.
Find $A_{i,d}$ in column $c$, say $A_{j,c}=A_{i,d}$ if it exists. 
Swap $A_{i,c}$ and $A_{i,d}$.
If $j$ is not found then stop; otherwise set $i=j$ and repeat. 
See an example in Fig. \ref{tableRS2}(b).
We have $c=1$ and $d=3$. 
Start with symbol $A_{4,1}=1$ which is not present in column $3$ in Fig. \ref{tableRS2}(a).
Function $f_{1,3}$ maps 1 to 4, 4 to 2. 
Swap $A_{4,1}$ and $A_{4,3}$, then swap $A_{1,1}$ and $A_{1,3}$.
The result is shown in Fig. \ref{tableRS2}(b).

Symbol cycles may occur in Latin rectangles. 
We also apply {\em symbol paths} which are defined as follows.
Consider distinct symbols $a$ and $b$ in a Latin rectangle $A$.
Start with a symbol $A_{i,j}=a$ in column $j$ such that symbol $b$ is not present in column $j$, 
so $(i,j)$ is the first cell in the symbol path. 
Repeat the following step.
Find symbol $\{a,b\}\setminus A_{i,j}$ in row $i$, say $A_{i,j'}$.
Append $(i,j')$ to the path. 
Find symbol $\{a,b\}\setminus A_{i,j'}$ in column $j'$.
If it does not exist, then stop.
Suppose it exists, say $A_{i',j'}$. Then append $(i',j')$ to the path,
set $i=i',j=j'$ and repeat the step. 
When the path is computed, swap symbols $a$ and $b$ in the path,
see an example in Fig. \ref{tableRS2}(c).

\begin{figure}[ht]
\begin{algorithmic}[1]
  \STATE {\bf Input:} $n \geq 5$ and $r\in [n,n^2]$ satisfying Theorem \ref{thm-pairs}.
  \STATE {\bf Output:} $n\times n$ Latin squares $A$ and $B$ such that $r(A,B)=r$.
  \STATE Compute two random permutations of $[n]$ and copy them to $A[1][*]$ and $B[1][*]$.
  \FOR{$k = 1, \ldots, n-1$}
\STATE For all $j\in [n]$, compute $A_j=[n]\setminus \{A[1][j],A[2][j],\dots,A[k][j]\}$.
\STATE For all $j\in [n]$, compute $B_j=[n]\setminus \{B[1][j],B[2][j],\dots,B[k][j]\}$.
\STATE Set $V_1=\{v_1,\dots,v_n\}$ and $V_2=\{u_1,\dots,u_n\}$.
\STATE Compute a bipartite graph $G_B=(V_1,V_2,E_B)$ 
where $E_B=\{(v_i,u_j)~|~u_j\in B_i\}$.
\STATE Compute a matching $M=\{(v_i,u_{b_i})~|~i\in [n]\}$ in $G_B$. 
\STATE For all $j\in [n]$, set $B[k+1][j]=b_j$.
\STATE Compute  $S=\{(A[i][j],B[i][j])~|~i\in [k], j\in [n]\}$.
\STATE Compute a graph $G=(V_1\cup V_2,E,w)$ where 
$E =\{(v_i,u_j)~|~j\in A_i\}$,
and $w(v_i,u_j) = 0$ if $(j,B[k+1][i]) \in S$; otherwise $w(v_i,u_j) = 1$.
\STATE Compute the maximum weight matching $M=\{(v_i,u_{a_i})~|~i\in [n]\}$ in $G$.
\STATE For all $j\in [n]$, set $A[k+1][j]=a_j$.
\ENDFOR
\WHILE{$r(A,B)\ne r$}
\STATE Compute $A'$ and $B'$ using a random row/column/symbol cycle 
\STATE or a random column/symbol path
\IF{$|r(A',B')-r|<|r(A,B)-r|$}
\STATE $A=A',B=B'$
\ENDIF
\ENDWHILE
\end{algorithmic}
\caption{Algorithm A2.}
\label{fig:A2}
\end{figure}

We implemented algorithm A2 using the method from Theorem \ref{tMaxRAB} combined with switching
row/ column/symbol cycles and column/symbol path, see the pseudocode in Fig. \ref{fig:A2}.
The results for $n=7,\dots,20$ are shown in Table \ref{t:res}. 
Note that the low values of $r(A,B)$ are covered completely (by Theorem \ref{thm-pairs}) and 
new values of $r(A,B)$ larger than the ones computed by algorithm A1.
The problem of computing all large values of $r(A,B)$ is quite difficult.
Specifically, the problem of computing random orthogonal Latin squares (i.e. for $r(A,B)=n^2$)
is very difficult.

\begin{table} 
\centering
\begin{tabular}{|r||c|c|c|} 
\hline
$n$ & A3 & A4 & $r(A,A^T)$\\
\hline
\hline
5 & 5,7,10-11,13-15,17,19,21,25 & - & 11.2769\\\hline 
6 & 6,8-31 & - & 15.4191 \\\hline 
7 & 7,9-45,47,49 & - & 20.2193\\\hline 
8 & 8,10-58 & 8,10-62,64 & 25.9692\\\hline 
9 & 20-69 & 9,11-79,81 & 32.368\\\hline 
10 & 29,33-84 & 10,12-98 & 39.5512\\\hline
11 & 42,44-99 & 11,13-119 & 47.4226\\\hline
12 & 56,59-117 & 12,14-140 & 56.1379\\\hline
13 & 72-133 & 13,15-163& 65.5507\\\hline
14 & 83-84,86,88-152 & 14,16-184 & 75.7708\\\hline
15 & 101,103-170,172-173 & 15,17-213 & 86.5367\\\hline 
16 & 118,120-196 & 16,18-242 & 98.443\\\hline
17 & 141-216 & 17,19-271 & 110.387\\\hline
18 & 161-239 & 18,20-307 & 123.722\\\hline
19 & 178-266,268,272 & 19,21-340 & 137.122\\\hline
20 & 202-292 & 20,22-375 & 152.056\\\hline 
\end{tabular}
\vspace{3mm}
\caption{
$r$-self-orthogonal Latin squares of order $n$ computed by Algorithm A3 (column A3) and Algorithm A4 (column A4).
No entry in column A4 means the same result as in column A3.
Column $r(A,A^T)$ shows the average value of $r(A,A^T)$ of a random Latin square $A$ using the program developed 
by Paul Hankin \cite{paulhankin19}.
The files for Latin squares computed by algorithm A4 are available in \cite{www}.}
\label{t:res2}
\end{table}

\section{Computing Self-Orthogonal Latin Rectangles}

We experiment with two algorithms using the algorithms developed in the previous section.
If $A$ is a $k\times n$ Latin rectangle then we assume that $B$ is an $n\times k$ Latin rectangle $B=A^T$.
For example, if a column cycle is applied to Latin rectangle $A$, it also affects $B$. 
We call the algorithms corresponding to algorithms A1 and A2, Algorithm A3 and A4, respectively.
The range of values of $r$ for $n\in [5,20]$ computed by Algorithm A3 
is shown in Table \ref{t:res2}. 
One can observe that the range is complete for $5\le n\le 9$  (by Theorem \ref{thm-self}).
In general,  algorithms A3 and A4 cover more values of $r$ compared to algorithms A1 and A2.

It is also related to an interesting problem:  What is the expected value of $r(A,A^T)$ where $A$ is a random Latin square of order $n$?
We computed approximately the expected value of $r(A,A^T)$ of a random Latin square $A$ 
using the program
developed by Paul Hankin \cite{paulhankin19} implementing an algorithm for generating random Latin squares
by Jacobson and Matthews \cite{jm-gud-96}. 
10000 runs were used for $n=5,\dots,15$ and 1000 runs were used for $n=16,\dots,20$.
The result is shown in Table \ref{t:res2} (see column $r(A,A^T)$).
For example, the expected value of $r(A,A^T)/n^2$ is close to 0.395 for $n=10$ and is close to 0.38 for $n=20$.

\section{Conclusions and Open Questions}
We developed randomized algorithms for computing pairs of $r$-orthogonal Latin squares of 
order $n$ and algorithms for computing $r$-self-orthogonal Latin squares of order $n$.
We implemented them and found that they cover many values of $r$ for $n$ in the range $5\le n\le 20$.
The problem of computing $r$-orthogonal Latin squares of  order $n$ and 
$r$-self-orthogonal Latin squares of order $n$ where $r$ is close to $n^2$ 
is difficult. 
For example, it is impossible to construct orthogonal Latin squares of order 
$n\ge 10$ and 
self-orthogonal Latin squares of order $n\ge 10$ by extending Latin rectangles minimizing $r(A,B)$ and $r(A,A^T)$, respectively. 
It would be interesting to find the expected values of $r(A,B)$ and $r(A,A^T)$ for random Latin squares $A$ and $B$ of order $n$.

\end{document}